\documentclass[letterpaper, 10 pt, journal, twoside]{IEEEtran}

\usepackage{xcolor}
\usepackage{verbatim}
\usepackage{soul}

\newcommand{\af}[1]{{\small\textsf{{#1}}}}
\def\G{\mathcal{G}}
\def\V{\mathcal{V}}
\def\E{\mathcal{E}}

\def\R{{\mathbb R}}
\def\N{{\mathcal N}}

\def\C{\mathcal{C}}
\def\H{\mathcal{H}}
\def\view{{\sf View}}
\def\EE{\mathbb{E}}
\def\cov{\text{Cov}}
\def\dist{\text{ dist}}
\def\L{\mathcal{L}} 
\def\f{\af{f}}
\def\r{\mathcal{R}}
\def\D{\af{D}}

\IEEEoverridecommandlockouts

\usepackage{amsmath}
\usepackage{amssymb,graphicx,verbatim,enumitem,mdframed,mathrsfs}
\usepackage{amsthm}
\usepackage{algorithm}
\usepackage{algpseudocode}
\usepackage{cite}

\newtheorem{theorem}{\bfseries Theorem}
\newtheorem{lemma}{\bfseries Lemma}
\newtheorem{corollary}{\textbf{Corollary}}
\newtheorem{definition}{\textbf{Definition}}

\providecommand{\mnorm}[1]{|#1|} 
\providecommand{\norm}[1]{\lVert#1\rVert}

\makeatletter
\newcommand\fs@betterruled{%
  \def\@fs@cfont{\bfseries}\let\@fs@capt\floatc@ruled
  \def\@fs@pre{\vspace*{5pt}\hrule height.8pt depth0pt \kern2pt}%
  \def\@fs@post{\kern2pt\hrule\relax}%
  \def\@fs@mid{\kern2pt\hrule\kern2pt}%
  \let\@fs@iftopcapt\iftrue}
\floatstyle{betterruled}
\restylefloat{algorithm}
\makeatother

\title{Preserving Statistical Privacy in \\Distributed Optimization}

\author{Nirupam Gupta$^\dagger$, Shripad Gade$^\star$, Nikhil Chopra$^\ddag$ and Nitin H. Vaidya$^\dagger$
\thanks{Partially supported by NSF through grant ECCS1711554, award 1842198, and ARL under Cooperative Agreement W911NF-17-2-0196.}
\thanks{$\dagger$ Department of Computer Science, Georgetown University, Washington, DC 20057, USA. $\star$ Electrical and Computer Engineering, University of Illinois Urbana-Champaign, Urbana, IL, USA. $\ddag$ Mechanical Engineering, University of Maryland, College Park, MD, USA. Email: {\em nirupam115@gmail.com}, {\em gade3@illinois.edu}, {\em nchopra@umd.edu}, and {\em nitin.vaidya@georgetown.edu}.}
}

\begin{document}

\maketitle
\thispagestyle{empty}

\begin{abstract}                
We present a distributed optimization protocol that preserves statistical privacy of agents' local cost functions against a passive adversary that corrupts some agents in the network. The protocol is a composition of a distributed ``{\em zero-sum}" obfuscation protocol that obfuscates the agents' local cost functions, and a standard non-private distributed optimization method. We show that our protocol protects the statistical privacy of the agents' local cost functions against a passive adversary that corrupts up to $t$ arbitrary agents as long as the communication network has $(t+1)$-vertex connectivity. The ``{\em zero-sum}" obfuscation protocol preserves the sum of the agents' local cost functions and therefore ensures accuracy of the computed solution.
\end{abstract}

\section{Introduction}
\label{sec:intro}

\IEEEPARstart{D}{istributed} optimization in multi-agent peer-to-peer networks has gained significant attention in recent years~\cite{yang2019survey}. In this problem, each agent has a local cost function and the goal for the agents is to collectively minimize sum of their local cost functions. Specifically, we consider $n$ agents, where each agent $i$ has a convex cost~$h_i: \mathbb{R}^m \to \mathbb{R}$ and a convex, compact set $\mathcal{X}$. A \emph{distributed optimization algorithm} enables the agents to collectively compute a {\em global} minimum,
\begin{equation}
    x^* \in \underset{x \in \mathcal{X}}{ \arg \min } \sum_{i = 1}^n h_i(x). \label{eqn:opt}
\end{equation}

We consider a scenario when a {\em passive} adversary can corrupt some of the agents in the network. The corrupted agents follow the prescribed protocol correctly, but may try to learn about the cost functions of other non-corrupted agents in the network. In literature, a passive adversary is also commonly referred as {\em honest-but-curious}. Prior work has shown that for certain distributed optimization algorithms, such as the Distributed Gradient Descent \af{(DGD)} method, a passive adversary may learn about all the agents' cost functions by corrupting only a subset of agents in the network~\cite{gade2016private}. This is clearly undesirable in general, and especially in cases where the cost functions may contain sensitive information~\cite{silaghi2004distributed}.

In this paper, we consider the Function Sharing \af{(FS)} protocol~\cite{gade2018private}, wherein the agents obfuscate their local cost functions with {\em correlated random functions} before executing a (non-private) distributed optimization algorithm such as the \af{DGD} method. The obfuscation strategy is aggregate invariant by construction and therefore, the agents compute a minimizer~\eqref{eqn:opt} accurately using solely their obfuscated local cost functions~\cite[Theorem~1]{gade2016private}. The \af{FS} protocol was first proposed by Gade et al.~\cite{gade2016private}. However, as of yet, the \af{FS} protocol lacks a formal privacy analysis. In this paper, we utilize the statistical privacy definition developed by Gupta et al.~\cite{gupta2017privacy,gupta2018privacy} to present a privacy guarantee of the \af{FS} protocol. 

In the past, distributed optimization protocols have been proposed for preserving differential privacy of the agents' local cost functions. However, these differetially private protocols suffer inevitably from privacy-accuracy trade-offs~\cite{nozari2018differentially, huang2015differentially}. That is, the agents can only compute an approximation of a global minimum $x^*$, defined by~\eqref{eqn:opt}.
The \af{FS} protocol allows the agents to compute a global minimum~\eqref{eqn:opt} accurately, and therefore, it obtains a weaker statistical privacy guarantee compared to the differentially private protocols.

Homomorphic encryption-based privacy protocols implicitly rely on two pragmatic assumptions, (1) computational intractability of {\em hard} mathematical problems, and (2) limited computational power of a passive adversary~\cite{silaghi2004distributed,hong2016privacy,LU2018314,8410603}. We show that the \af{FS} protocol provides {\em statistical} (or information-theoretic~\cite{katz2014introduction}) privacy, which is valid regardless of the above assumptions.

However, both the differetial privacy based protocols and the homomorphic encryption based protocols can provide privacy against eavesdroppers~\cite{nozari2018differentially, huang2015differentially, silaghi2004distributed, hong2016privacy,LU2018314,8410603}. The \af{FS} protocol, on the other hand, can only provide privacy against honest-but-curious agents in the network.

\vspace{0.3em}
\noindent{\bf \em Summary of Our Contributions:} We show that in the \af{FS} protocol the passive adversary obtains limited information, in a statistical sense, about the local cost functions of the non-corrupted (or {\em honest}) agents, as long as the agents corrupted by the passive adversary do not form a {\em vertex cut} in the underlying communication network topology. Thus, the \af{FS} protocol protects the statistical privacy of the honest agents' local cost functions against any passive adversary that corrupts up to $t$ arbitrary agents in the system as long as the communication network topology has $(t+1)$-vertex connectivity. 

It is of independent interest to note that a variant of the \af{FS} protocol is known to preserve the {\em perfect} statistical privacy in distributed average consensus problem~\cite{gupta2018information, gupta2019statistical,gade2020private}.

\section{Problem Setup}
\label{sec:prob_f}

\def\A{\mathcal{A}}
We consider a passive adversary, denoted by $\A$, that corrupts some agents in the network. The goal is to design distributed optimization protocols that protect the privacy of the non-corrupted (or {\em honest}) agents' local cost functions against the passive adversary, while allowing the agents to compute solution~\eqref{eqn:opt} accurately. The adversary is assumed passive and the corrupted agents execute the prescribed protocol correctly. 
For a distributed optimization protocol $\Pi$, we define {\em view} of $\A$ for an execution of $\Pi$ as follows.


\begin{definition}
\label{def:view}
For a protocol $\Pi$, the {view} of $\A$ constitutes the information {stored}, {transmitted} and {received} by the agents corrupted by $\A$ during the execution of $\Pi$.
\end{definition}


Privacy requires that the entire \emph{view} of $\A$
does not leak significant (or any) information about the local costs of the honest agents. Note that, by definition, $\A$ inevitably learns a point $x^* \in \arg \min_{x \in \mathcal{X}} \sum_{i = 1}^n h_i(x)$, assuming it corrupts at least one agent. A {\em perfectly} private protocol would not reveal any information about the honest agents' cost functions to $\A$ besides $x^*$. However, such a perfect privacy is quite difficult to guarantee. For now, we relax the privacy requirement, and only consider privacy for the {\em affine} terms of the agents' cost functions. However, as elaborated in Section~\ref{sec:extension}, the {\footnotesize \textsf{FS}} protocol can be extended easily for privacy of {\em higher-order polynomial} terms. That is, we implicitly assume that the non-affine terms of the agents' cost functions are known a priori to the passive adversary.

For each agent $i$, the cost function $h_i(x)$ can be decomposed into two parts; the {\em affine} term denoted by $h_i^{(1)}(x)$, and the {\em non-affine} term denoted by $h_i^{\dagger}(x)$. Specifically,
\begin{align}
    h_i(x) = h^{(1)}_i(x) + h^{\dagger}_i(x), \quad \forall x \in \R^m, \, i \in \{1, \ldots, \, n\}. \label{eqn:parts}
\end{align}
As the name suggests, the affine terms are affine in $x$. That is, for each $i$ there exists $\alpha_i \in \R^m$ and $\gamma_i \in \R$ such that, $h_i^{(1)}(x) = \alpha^T_i \, x + \gamma_i, \; \forall x \in \R^m$,
where $(\cdot)^T$ denotes the transpose. As the constants $\gamma_i$'s do not affect the solution of the optimization problem~\eqref{eqn:opt}, the agents need not share these constants with each other. Hence, the privacy of honest agents' $\gamma_i$'s can be trivially preserved. For a meaningful discussion of privacy we will ignore these constants. Let,
\begin{align}
    A = \left[\alpha_1, \ldots, \, \alpha_n\right] \label{eqn:alpha}
\end{align}
be the $m \times n$-dimensional matrix obtained by column-wise stacking of the individual agents' {affine coefficients}. 

Let $\C$ denote the set of agents corrupted by the adversary $\A$, and let $\H$ denote the remaining non-corrupted (or honest) agents. For privacy preservation, the protocol $\Pi$ may introduce some randomness in the system, in which case the {\em view} of $\A$ is a random variable. 
Let,

\begin{itemize}
    \item $\view_\C(A)$ denote the probability distribution of the {\em view} of $\A$ for an execution of $\Pi$ when the agents' private cost functions have affine coefficients~$A$.
\end{itemize}


Our definition of privacy below is built on {\em relative entropy}, which is also known as the Kullback-Leibler (KL) divergence. For a continuous random variable $\r$, let $\f_\r(r)$ denote its probability distribution or probability density function (p.d.f.) at $r \in \r$. The KL-divergence, denoted by $\D_{KL}$, quantifies the difference between a certain probability distribution $\f'_\r$ and the reference probability distribution $\f_\r$~\cite{cover2012elements}. Specifically, the KL-divergence of $\f'_{\r}$ from $\f_{\r}$ is defined as
\[\D_{KL} \left(\f_{\r}, \f'_{\r} \right) = \int_{\mathcal{R}}\f_{\r}(s) \log \left(\frac{\f_{\r}(s)}{\f'_{\r}(s)} \right) ds.\]
Let $\norm{.}$ denote the Euclidean norm for vectors and the Frobenius norm for matrices. 


\def\a{{(a)}}


\begin{definition}\label{def:ip}
For $\epsilon > 0$, a distributed optimization protocol $\Pi$ is said to be \mbox{``$(\C, \epsilon)$-affine private''} if for every pair of agents' affine coefficients $A = [\alpha_1, \ldots, \, \alpha_n]$ and $B = [\beta_1, \ldots, \, \beta_n]$ subject to the constraints:
\begin{align}
    \alpha_i = \beta_i , \; \forall \, i \in \C, \; \text{ and } \;
    \sum_{i \in \H} \alpha_i = \sum_{i \in \H} \beta_i , \label{eqn:constraint}
\end{align}%
the supports of $\view_\C(B)$ and $\view_\C(A)$ are identical, and 
\begin{align}
    \D_{KL} \left(\view_\C(A), \,\view_\C(B) \right) \leq \epsilon \norm{A - B}^2. \label{eqn:def_d_kl}
\end{align}
\end{definition}

In other words, Definition~\ref{def:ip} implies that if $\Pi$ is \mbox{\em $(\C, \epsilon)$-affine private} then an adversary $\A$ cannot unambiguously distinguish between two sets of agents' affine coefficients, $A$ and $B$, that are identical for the corrupted agents and have identical sum over all honest agents (i.e., satisfy~\eqref{eqn:constraint}). The value of $\epsilon$ signifies the {\em strength} of the privacy obtained. Smaller is the value of $\epsilon$, the more difficult it is for $\A$ to distinguish between two sets of agents' affine coefficients satisfying~\eqref{eqn:constraint}, and hence stronger is the privacy.




\section{Proposed Protocol and Privacy Guarantee}
\label{sec:pm}
In this section, we present the Function Sharing (\af{FS}) protocol and the formal privacy guarantee.

The notation used is as follows. The underlying communication network is modeled by an undirected graph~$\G = (\V, \, \E)$, where the set of nodes $\V = \{1, \ldots, \, n\}$ denotes the agents (indexed arbitrarily), and the set of edges $\E$ denotes the communication links between the agents. Being undirected, each edge $e \in \E$ is represented by an unordered pair of agents. For each $i$, the set $\N_i = \left\{ j \in \V ~ \vline ~ \left\{ i, \, j \right\} \in \E \right\}$ denotes the neighbors of agent $i$.

The \af{FS} protocol constitutes two phases as elaborated in Algorithm~\ref{Algo:FS-DGD}. In phase I, each agent~$i$ uses a ``zero-sum'' obfuscation protocol to compute an ``{\em effective cost function}'' $\widetilde h_i(x)$ based on its private local cost function~$h_i(x)$. In phase II, the agents use the \af{DGD} algorithm on their effective local cost functions to solve for the {\em effective} optimization problem,
\begin{align}
    \underset{x \in \mathcal{X}}{\text{minimize}} \sum_{i = 1}^n \widetilde h_i(x). \label{eqn:opt_new}
\end{align}

We now show that upon completion of phase II the agents indeed obtain a common minimum of the original optimization problem~\eqref{eqn:opt}.
As $\G$ is an undirected graph, 
$$\sum_{i = 1}^n u_i = \sum_{i = 1}^n \sum_{j \in \N_i}(r_{ji}-r_{ij}) = 0.$$

\noindent This implies that, for all $x \in \R^m$,
\begin{align}
    \sum_{i = 1}^n \widetilde h_i(x)  =  \sum_{i = 1}^n h_i(x) + \sum_{i = 1}^n u_i^T x = \sum_{i = 1}^n h_i(x). \label{eqn:sum_preserve}
\end{align}
Equivalently, the masking in phase I preserves the sum of the agents' local cost functions. Therefore, a solution for problem~\eqref{eqn:opt_new}, obtained using the \af{DGD} algorithm in~\cite{gade2016distributed}, is a solution for the original optimization problem~\eqref{eqn:opt}.





\begin{algorithm}[t!]
\caption{Function Sharing (\af{FS}) Protocol\label{Algo:FS-DGD}}
\begin{algorithmic}[1]
\Statex \textbf{Input: }Each agent $i$ has cost function $h_i(x)$ and $\sigma \in \mathbb{R}$.
\Statex \textbf{Output: }Minimizer, $x^* \in \arg \min_{x \in \mathcal{X}} \sum_{i=1}^n h_i(x)$
\Statex \hspace{-0.2in} \textbf{$\Diamond$ Phase 1:} {\em Masking of Cost Functions} 
\Statex Each agent $i\in\mathcal{V}$ executes:
\State Draws vectors $r_{ij} \sim N\left( 0_m, \,  \sigma^2 I_m\right)$ independently for $j \in \mathcal{N}_i$ and sends $r_{ij}$ to each agent~$j \in \N_i$. 
\State Compute the {\em mask} $u_i$ 
\begin{equation}
    u_i = \sum_{j \in \N_i}(r_{ij}-r_{ji}) \label{eqn:i_vn}
\end{equation}
\State Compute the {\em effective cost function} $\widetilde h_i(x)$,
\begin{align}
    \widetilde{h}_i (x) = h_i (x) + u_i^T x, \quad \forall x \in \R^m. \label{eqn:eff_cost}
\end{align}
\Statex \hspace{-0.2in}\textbf{$\Diamond$ Phase 2:} {\em Distributed Optimization} 
\State Agents execute the \af{DGD} algorithm~\cite{nedic2009distributed} on the local effective costs $\{\widetilde{h}_i(x)\}_{i\in\V}$. 
\end{algorithmic}%
\end{algorithm}%
\subsection{Privacy Guarantee}
\label{sub:pm_priv}

The privacy guarantee for the above \af{FS} protocol is presented by Theorem~\ref{thm:priv_2} below. 
Recall that $\C$ denotes the set of agents corrupted by the passive adversary. Let $\H = \V \setminus \C$ denote the set of honest agents, and let $\G_\H$ denote the residual graph obtained by removing the agents in $\C$, and the edges incident to them, from $\G$. Let $\L_{\H}$ denote the graph-Laplacian of $\G_\H$ and $\underline{\mu}(\L_{\H})$ denote the second smallest eigenvalue of $\L_\H$. The eigenvalue $\underline{\mu}(\L_{\H})$ is also commonly known as the {\em algebraic connectivity} of the graph~\cite{godsil2001algebraic}. 


\begin{theorem}
\label{thm:priv_2}
If $\C$ is not a vertex cut of $\G$, and the affine coefficients of the agents' private cost functions are independent of each other, then the \af{FS} protocol is $(\C, \epsilon)$-affine private, with $\epsilon = 1/(4\sigma^2 \underline{\mu}(\L_{\H}))$. 
\end{theorem}

Theorem~\ref{thm:priv_2} implies that $\C$ not being a vertex cut\footnote{A {\em vertex cut} is a set of vertices of a graph which, if removed -- together with any incident edges -- disconnects the graph~\cite{godsil2001algebraic}.} of $\G$ is sufficient for $(\C, \epsilon)$-affine privacy. Note that, smaller the value of $\epsilon$, stronger is the privacy. According to Theorem~\ref{thm:priv_2}, $\epsilon$ is inversely proportional both to the variance $\sigma^2$ of the elements of random vectors $r_{ij}$'s used for masking of agents' local costs, and the {\em algebraic connectivity} of the residual network topology $\G_\H$. Therefore, the agents can achieve stronger privacy by using random vectors with larger variances (i.e., larger $\sigma^2$) in phase I of the \af{FS} protocol. Additionally, \af{FS} protocol guarantees stronger privacy if the residual honest graph $\G_\H$ is densely connected. 



We further note that the \af{FS} protocol can guarantee privacy against any passive adversary that corrupts at most $t$ agents in the network if the network has $(t+1)$-vertex connectivity. Specifically, we have the following corollary of Theorem~\ref{thm:priv_2}.


\begin{corollary}
\label{cor:f}
If $\G$ has $(t+1)$-vertex connectivity and the affine coefficients of the agents' private cost functions are independent of each other, then for an arbitrary set $\C \subseteq \V$ with $|\C| \leq t$ the \af{FS} protocol is $(\C,\epsilon)$-affine private with
\[ \epsilon = \max \left\{ \frac{1}{4\sigma^2 \underline{\mu}(\L_{\H})} ~ \vline ~ \H = \V \setminus \C, ~ \mnorm{\C} \leq t \right\}. \]
\end{corollary}


The above connectivity condition for privacy is indeed tight. Specifically, the $(t+1)$-vertex connectivity is {\em necessary} for privacy against at most $t$ colluding honest-but-curious agents in the consensus-based distributed gradient and subgradient optimization algorithms~\cite{gade2016private,gade2017private, yan2013distributed}.

\subsection{Privacy of Higher-Degree Polynomial Terms}
\label{sec:extension}

The \af{FS} protocol presented in Algorithm~\ref{Algo:FS-DGD} only protects the privacy of affine coefficients of local cost functions, as formally stated in Theorem~\ref{thm:priv_2}. In what follows, we show an easy extension to protect privacy of higher degree polynomial terms of agents' private cost functions. Here, we assume the agents' cost functions to be univariate, i.e., $x \in \R$.

For each agent $i$, let $\alpha^{(\ell)}_i$ denote the $\ell$-th degree coefficient of its cost function $h_i(x)$.
Similar to the definition of {\em $(\C, \, \epsilon)-$affine privacy}, we now define the privacy of the $\ell$-th degree coefficients $A^{(\ell)} = [\alpha^{(\ell)}_1, \ldots, \alpha^{(\ell)}_n]$ against a passive adversary that corrupts a set of agents $\C$. Let $\view_\C(A^{(\ell)})$ denote the probability distribution of the {\em view} of adversary $\A$ when $\ell$-th degree coefficients of agents' private cost functions are given by $A^{(\ell)}$.  


\vspace{0.6em}

\noindent{\em Privacy Definition}: For $\epsilon > 0$, protocol $\Pi$ is said to preserve the \mbox{\em $(\C, \epsilon)$-privacy} of $\ell$-th degree coefficients $A^{(\ell)}$ if for every other set of $\ell$-th degree coefficients $B = [\beta^{(\ell)}_1, \ldots, \, \beta^{(\ell)}_n]$ subject to the constraints: $$\beta^{(\ell)}_i = \alpha^{(\ell)}_i, \; \forall \, i \in \C, \; \text{ and } \;
    \sum_{i \in \H} \beta^{(\ell)}_i = \sum_{i \in \H} \alpha^{(\ell)}_i,$$
the support of $\view_\C(A^{(\ell)})$ \& $\view_\C (B^{(\ell)})$ are identical, and
\[\D_{KL} \left(\view_\C(A^{(\ell)}), \,\view_\C (B^{(\ell)}) \right) \leq \epsilon \norm{A^{(\ell)} - B^{(\ell)}}^2. \label{eqn:def_d_kl_poly}\]

When defining the distribution $\view_\C(A^{(\ell)})$, we implicitly assume that the passive adversary $\A$ knows all the coefficients of the honest agents' costs, except the $\ell$-th coefficients $\{\alpha^{(\ell)}_i, ~ i \in \H\}$. Thus, the privacy analysis here is conservative.

\vspace{0.6em}

\noindent{\em Modified \af{FS} Protocol and Privacy Guarantee}: In the first phase, the agents mask the coefficients $A^{(\ell)}$ in a similar manner as the masking of the affine coefficients delineated in Algorithm~\ref{Algo:FS-DGD} to compute the effective cost functions.

Note that in this case, due to the non-affine masking, the effective cost functions $\widetilde{h}_i(x)$'s may become non-convex. The sum of the effective cost functions, however, is still a convex function (see~\eqref{eqn:sum_preserve}). 
As discussed in~\cite{gade2016distributed}, the \af{DGD} algorithm allows agents to minimize convex sum of their local non-convex cost functions, provided that the local cost functions' gradients are Lipschitz continuous~\cite[Theorem 1]{gade2016private}. The \af{DGD} can be substituted with other distributed optimization algorithms, provided those algorithms also minimize convex sum of non-convex functions (see \cite{gade2016distributed} for details).

Now, Theorem~\ref{thm:priv_2} implies that if $\C$ does not form a vertex cut of the network topology $\G$ then the \af{FS} protocol, modified as above, preserves the \mbox{\em $(\C, \epsilon)$-privacy} of $\ell$-th degree coefficients $A^{(\ell)}$ for each $\ell = \{1,\ldots,d\}$, where, privacy parameter $\epsilon = 1/(4\sigma^2 \underline{\mu}(\L_{\H}))$.

\section{Proof of Theorem~\ref{thm:priv_2}}
\label{sec:pa}
In this section, we present the formal proof for Theorem~\ref{thm:priv_2}. In principle, the proof is a generalization of the privacy analysis presented in~\cite{gupta2017privacy}. First, we state a few critical observations in Lemmas~\ref{lem:dist_a} and~\ref{lem:priv_1} below.


Let $\L$ denote the graph-Laplacian of the network topology $\G$.  
As $\G$ is undirected, $\L$ is a diagonalizable matrix~\cite{godsil2001algebraic}. Specifically, there exists a unitary matrix $M$ constituting the orthogonal eigenvectors of $\L$ such that\footnote{$Diag(y_1, .. , y_n)$ is a diagonal matrix with diagonal entries $y_1, .. , \, y_n$.}, $\L = M Diag\left(\mu_1, \ldots, \, \mu_{n} \right)M^T$ where $\mu_1 \leq \mu_2 \leq \cdots \leq \mu_n$ are the eigenvalues of $\L$. When $\G$ is connected, $\mu_1 = 0$ and $\mu_2 > \mu_1$~\cite{godsil2001algebraic}. We denote the {\em generalized inverse} of $\L$ by $\L^{\dagger}$. Note that~\cite{gutman2004generalized},
\begin{align}
    \L^{\dagger} = M~Diag \left(0, \, 1/\mu_2, \ldots, \, 1/\mu_{n}\right) ~ M^T \label{eqn:gen_inv}
\end{align}
For future usage, we denote the second smallest eigenvalue of $\L$, i.e., $\mu_2$, by $\underline{\mu}(\L)$. Let $0_n$ and $1_n$ denote the zero and the one vectors, respectively, of dimension $n$. For a positive real value $c$, $N^\dagger(0_n, \, c \L)$ denotes the {\em degenerate} Gaussian distribution~\cite{rao1973linear}. Specifically, 
if $\r \sim N^\dagger(0_n, \, c \L)$ and $\G$ is a connected graph then,
\begin{align}
     f_{\r}(r) = \begin{cases}
     \frac{1}{\sqrt{\det^*(2\pi c \L)}}\exp\left(-\frac{r^T\L^{\dagger}r}{2c}\right) & \text{, } r^T 1_n = 0 \\
     0 & \text{, otherwise} 
     \end{cases}
\end{align}
where ${\det}^* (2 \pi c \L) = (2 \pi c)^{n-1}\prod_{i = 2}^{n}\mu_i$. Henceforth, for a vector $v$, $v^k$ denotes the $k$-th element of $v$ unless otherwise noted. For $i \in \V$, recall that $u_i$ is the {\em mask} (see~\eqref{eqn:i_vn}). Let,
\begin{align}
    U^k = \left[ u^k_1, \ldots, u^k_n \right]^T, \quad k = 1, \ldots, \, m. \label{eqn:all_masks}
\end{align}
be a $n$-dimensional vector comprising the $k$-th elements of the {\em masks} computed by the agents in phase I of the \af{FS} protocol. For a random vector $\r$, we denote its mean by $\EE(\r)$ and its covariance matrix by $\cov(\r)$. Note that \\
$\cov(\r) = \EE \left(\r - \EE(\r) \right) \left(\r - \EE(\r) \right)^T$.




\begin{lemma}
\label{lem:dist_a}
If $\G$ is a connected graph then for each $k \in \{1, \ldots, \, m\}$, $U^k \sim N^\dagger \left( 0_n, ~ \, 2\sigma^2 \L \right)$.
\end{lemma}


\begin{IEEEproof}
Assign an arbitrary order to the set of edges, i.e., let $\E = (e_1, \ldots, \, e_{\mnorm{\E}})$.
For each edge $e_l$ where $l \in \{1, \ldots, \, \mnorm{\E}\}$, we define a vector $\theta_{e_l}$ of size $n$ whose $i$-th element denoted by $\theta^i_{e_l}$ is given as follows:
\[
	\theta^i_{e_l} = \left\{\begin{array}{cl}1 & \hspace*{3pt} \text{if } e_l = \{i,\,j\} \text{ and } i < j\\ -1 &  \hspace*{3pt} \text{if } e_l = \{i,\,j\} \text{ and } i > j \\ 0 &  \hspace*{3pt} \text{otherwise.}\end{array}\right. 
\]
Let $\Theta = \left[ \theta_{e_1}, \ldots, \,\theta_{e_{\mnorm{\E}}} \right]$ be an \emph{oriented incidence matrix} of graph $\G$~\cite{godsil2001algebraic}.
For each edge $e = \{i , \, j \}$ with $i<j$, 
\begin{align}
    c_{e} \triangleq r_{ji}-r_{ij}. \label{eqn:b_e}
\end{align}
Since the each random vector in $\{r_{ij}, ~ i, \, j \in \V\}$ is identically and independently distributed (i.i.d.) by a normal distribution $N(0_m, \sigma^2 I_m)$,~\eqref{eqn:b_e} implies that for each edge $e_l$ the random vector $c_{e_l}$ is i.i.d.~as $N(0, 2 \sigma^2 I_m)$. Therefore, for each $k$, the random variable $c^k_{e_l}$ has normal distribution of $N(0, 2 \sigma^2)$. Let, $C^k = [c^k_{e_1}, c^k_{e_2},  \cdots, c^k_{e_{\mnorm{\E}}}]^T$.
For two distinct edges $e$ and $e'$, the random vectors $c_{e}$ and $c_{e'}$ are independent. Therefore, 
\begin{align}
    \EE (C^k) (C^k)^T = 2 \sigma^2 I_{\mnorm{\E}}, \label{eqn:cov_bk}
\end{align}
where $I_{\mnorm{\E}}$ is $\mnorm{\E} \times \mnorm{\E}$ identity matrix. Moreover, from~\eqref{eqn:i_vn}, $U^k = \Theta \, C^k, \; \forall k \in \{1, \ldots, \, m\}$. As $\G$ is assumed connected, the support of $U^k$ is the entire space orthogonal to $1_n$. Also, $\EE(U^k) = \Theta \EE(C^k) = 0_n$. As $\L = \Theta \Theta^T$~\cite{godsil2001algebraic}, $\cov(U^k) = \Theta \left (\EE (C^k) (C^k)^T \right) \Theta^T = 2 \sigma^2 \, \Theta \Theta^T = 2\sigma^2 \L$. Thus, $U^k$ has the generalized Gaussian distribution $N^\dagger(0_n, \, 2 \sigma^2 \L)$. 
\end{IEEEproof}


Using the above lemma, we show that the knowledge of the {\em effective cost functions} does not provide significant information about the {\em affine coefficients} of the agents' private cost functions. 

Consider two possible executions $E_A$ and $E_B$ of the \af{FS} protocol such that the affine coefficients of the agents' effective cost functions in both executions are given by $\widetilde{A} = \left[\widetilde{\alpha}_1, \ldots, \, \widetilde{\alpha}_n \right]$.
In execution $E_A$, the agents have local cost functions with affine coefficients $A = [\alpha_1, \ldots, \alpha_n]$, and in execution $E_B$, the agents have local cost functions with affine coefficients $B = [\beta_1, \ldots, \beta_n]$. Let $f_{\widetilde{A}|A}$ and $f_{\widetilde{A}|B}$ denote the conditional p.d.f.s~of $\widetilde{A}$ given that the affine coefficients of the agents' private cost functions are $A$ and $B$, respectively. Recall that $\underline{\mu}(\L)$ denotes the second smallest eigenvalue of the graph-Laplacian matrix $\L$, i.e., $\mu_2$.



\begin{lemma}
\label{lem:priv_1}
If $\mathcal{G}$ is connected, and $\sum_{i = 1}^n \alpha_i = \sum_{i = 1}^n \beta_i$, then supports of $f_{\widetilde{A}|A}$ and $f_{\widetilde{A}|B}$ are identical, and 
\begin{align}
    \D_{KL} \left(f_{\widetilde{A}|A}, \, f_{\widetilde{A}|B} \right) \leq \frac{1}{4\sigma^2\underline{\mu}(\L)} \norm{A - B}^2 ~ .
\end{align}
\end{lemma}

\begin{IEEEproof}
Let, $\widetilde{A}^k$ and $A^k$ denote the column vectors representing the $k$-th rows of the effective affine coefficeints $\widetilde{A}$ and the actual affine coefficients $A$, respectively. That is, $\widetilde{A}^k = \left[ \widetilde{\alpha}^k_1, \ldots, \, \widetilde{\alpha}^k_n\right]^T \text{ and } A^k = \left[ \alpha^k_1, \ldots, \, \alpha^k_n\right]^T$.
The proof comprises three parts.

\noindent {\em Part I:} Recall from~\eqref{eqn:i_vn}, $\widetilde{\alpha}^k_i = \alpha^k_i + u^k_i$ for all $i$ and $k$. 
Therefore (see~\eqref{eqn:all_masks} for the notation $U^k$), $\widetilde{A}^k = A^k + U^k$.
As $U^k$ is independent of $A^k$ for every $k$, we get,
\begin{align}
    f_{\widetilde{A}^k | A^k}(\widetilde{\alpha}^k_1, \ldots, \, \widetilde{\alpha}^k_n) = 
    f_{U^k}\left( \widetilde{A}^k - A^k   \right). \label{eqn:alpha_a_k}
\end{align}
Therefore, from Lemma~\ref{lem:dist_a}, if $\sum_{i = 1}^n \widetilde{\alpha}^k_i = \sum_{i = 1}^n \alpha^k_i$ then,
\begin{align}
\begin{split}
    & f_{\widetilde{A}^k | A^k}(\widetilde{\alpha}^k_1, \ldots, \, \widetilde{\alpha}^k_n) = \\
    & \frac{1}{\sqrt{\det^*(4\pi\sigma^2\L)}}\exp\left(-\frac{ (\widetilde{A}^k - A^k) ^T\L^{\dagger} (\widetilde{A}^k - A^k)}{4\sigma^2}\right) \label{eqn:pdf_1}
\end{split}
\end{align}
Else if $\sum_{i = 1}^n \widetilde{\alpha}^k_i \neq \sum_{i = 1}^n \alpha^k_i$ then
\begin{align}
    f_{\widetilde{A}^k | A^k}(\widetilde{\alpha}^k_1, \ldots, \, \widetilde{\alpha}^k_n) = 0,  \label{eqn:pdf_2}
\end{align}
From~\eqref{eqn:pdf_1} and~\eqref{eqn:pdf_2}, it is easy to see that the supports of the conditional p.d.f.s $f_{\widetilde A^k | A}$ and $f_{\widetilde A^k | B}$ are identical. 


\noindent{\em Part II}: From~\eqref{eqn:pdf_1}, 
\begin{align*} 
    \log \frac{f_{\widetilde A^k | A^k}(\widetilde{\alpha}^k_1, \ldots, \, \widetilde{\alpha}^k_n)}{f_{\widetilde A^k | B^k}(\widetilde{\alpha}^k_1, \ldots, \, \widetilde{\alpha}^k_n)} = \frac{(A^k - B^k)^T \L^{\dagger} (2\widetilde{A}^k - A^k - B^k)}{4 \sigma^2}
\end{align*}
Let $s = \widetilde{A}^k-A^k$, then we get, $\D_{KL}\left( f_{\widetilde{A}^k | A^k}, \, f_{\widetilde{A}^k | B^k} \right) = $
\begin{small}\begin{align*}
    &\frac{1}{4\sigma^2}\int_{s \in \R^n}(A^k - B^k)^T \L^{\dagger} (2s + A^k - B^k)f_{U^k}(s) \, ds \\
    &=\frac{1}{2\sigma^2}(A^k - B^k)^T\L^{\dagger}\EE(U^k) + \frac{1}{4\sigma^2}(A^k -  B^k)^T \L^{\dagger} (A^k - B^k).
\end{align*}%
\end{small}%
From Lemma~\ref{lem:dist_a}, $\EE(U^k) = 0_n$. Therefore, 
\begin{align}
    \D_{KL}\left(f_{\widetilde{A}^k | A^k} , \, f_{\widetilde{A}^k | B^k} \right) = \frac{1}{4\sigma^2}(A^k - B^k)^T \L^{\dagger} (A^k - B^k). \label{eqn:we_get}
\end{align}
As $\G$ is assumed connected, $\text{rank}(\L) = n-1$ and $\L 1_n = 0_n$. Recall that $1^T_n(A^k - B^k) = 0_n$. Thus, the vector $A^k - B^k$ belongs to the space orthogonal to the nullspace of $\L$. Now, substituting $\L^{\dagger}$ from~\eqref{eqn:gen_inv} in~\eqref{eqn:we_get} we obtain that
\begin{align}
    \D_{KL} \left(f_{\widetilde{A}^k | A^k} , \, f_{\widetilde{A}^k | B^k} \right) \leq \frac{\norm{A^k - B^k}^2}{4\sigma^2\underline{\mu}(\L)}.  \label{eqn:d_kl_1}
\end{align}
\noindent{\em Part III}: For $k \neq l$, $U^k$, $U^{l}$ are independent of each other. From~\eqref{eqn:alpha_a_k}, $f_{\widetilde{A} | A} = \prod_{k = 1}^m f_{\widetilde{A}{^k} | A^k} ~ \text{, and similarly, } ~ f_{\widetilde{A} | B} = \prod_{k = 1}^m f_{\widetilde{A}{^k} | B{^k}}.$
This, due to the KL-divergence property, implies that 
\[\D_{KL}\left( f_{\widetilde{A} | A}, \, f_{\widetilde{A} | B} \right) = \sum_{k = 1}^m \D_{KL}\left(f_{\widetilde{A}^k | A^k}, \, f_{\widetilde{A}^k | B^k} \right).\]
Substituting from~\eqref{eqn:d_kl_1} above concludes the proof.
\end{IEEEproof}

\vspace{0.4em}

Theorem~\ref{thm:priv_2} can be now proved easily using Lemma~\ref{lem:priv_1}.

\vspace{0.4em}
\noindent \textbf{\em Proof of Theorem~\ref{thm:priv_2}. }Recall that $\C$ denotes the set of corrupted agents and $\H=\V\setminus \C$ denotes the set of honest agents. Let $\E_\C$ denote set of edges incident to~$\C$ and $\E_{\H} = \E \setminus \E_\C$ be the set of edges incident only to honest agents.

Let the agents' true affine coefficients be given by an $m \times n$-dimensional matrix $A = [\alpha_1, \ldots, \, \alpha_n]$, as defined in~\eqref{eqn:alpha}. Recall the definition of $\view_\C(A)$ from Section~\ref{sec:prob_f}. In this part, we derive the p.d.f.~of $\view_\C(A)$ for the \af{FS} protocol, assuming the worst-case scenario where the {\em effective cost functions} of all the agents are revealed to the corrupted agents in the second phase. From Definition~\ref{def:view}, note that the {\em view} of the adversary $\A$ for the \af{FS} protocol comprises the following information:
\begin{enumerate}[leftmargin=*]
    \item The corrupted agents' private and {\em effective} cost functions, i.e., $\{h_i(x), \, \widetilde{h}_i(x), ~ i \in \C\}$.
    \item The set of random vectors $R_\C = \{r_{ij}, ~ \{i, \, j\} \in \E_\C\}$. 
    \item The {\em effective cost functions} of the honest agents, i.e., $\{\widetilde{h}_i(x), ~ i \in \H\}$. 
\end{enumerate}

For each agent $i \in \V$, let $\widetilde{\alpha}_i$ denote the affine coefficient of $\widetilde{h}_i(x)$. Let $\widetilde{A}_\C = [\widetilde{\alpha}_i, ~ i \in \C]$ and $\widetilde{A}_{\H} = [\widetilde \alpha_i , ~ i \in \H]$ be the collection of the effective affine coefficients of the corrupted and the honest agents, respectively. Let $f_{\left(\widetilde{A}_\H, \,\widetilde{A}_\C,  \, R_\C \right) \, \vline \,   A}$ denote the conditional joint p.d.f.~of $\widetilde{A}_\H$, $\widetilde{A}_\C$ and $R_\C$ given the agents' true affine coefficients $A$. From above we obtain that
\begin{align}
    \view_\C(A) = f_{\left(\widetilde{A}_\H, \,\widetilde{A}_\C,  \, R_\C \right) \, \vline \, A}~. \label{eqn:view_final}
\end{align}
For each agent $i \in \V$, let $\C_i = \N_i \cap \C$. Note that, see~\eqref{eqn:eff_cost}, 
\begin{align}
    \widetilde{\alpha}_i = \alpha_i + \sum_{j \in \N_i \setminus \C_i} (r_{ij} - r_{ji}) + \sum_{j \in \C_i} (r_{ij} - r_{ji}), ~ \forall i. \label{eqn:alpha_alpha}
\end{align}
For each honest agent $i \in \H$, let 
\begin{align}
    \overline{\alpha}_i = \alpha_i + \sum_{j \in \N_i \setminus \C_i} (r_{ij} - r_{ji}). \label{eqn:over_alpha}
\end{align}
Let $\overline{A}_{\H} = \left[ \overline{\alpha}_i, ~ i \in \H\right]$ be the collection of honest agents' $\overline{\alpha}_i$'s. Recall that, for two agents $i$ and $j$, the vectors $r_{ij}, ~ r_{ji} \in R_\C$ if and only if $i \in \C$ or $j \in \C$. Therefore, for each honest agent $i \in \H$, the value of $\sum_{j \in \C_i} (r_{ij} - r_{ji})$ is deterministic given $R_\C$. Thus,
\begin{align}
    f_{\left(\widetilde{A}_\H, \,\widetilde{A}_\C,  \, R_\C \right) \, \vline \,  A} = f_{\left(\overline{A}_{\H}, \,\widetilde{A}_\C,  \, R_\C \right) \, \vline \,  A}. \label{eqn:prob_view}
\end{align}
As the agents' affine coefficients are assumed independent of each other, we have from~\eqref{eqn:over_alpha}, $\overline{A}_{\H}$ is independent of $\widetilde{A}_{\C}$. Moreover,~\eqref{eqn:over_alpha} also implies that $\overline{A}_{\H}$ is independent of $R_\C$. Therefore, $f_{\left(\overline{A}_{\H}, \,\widetilde{A}_\C,  \, R_\C \right) \, \vline \,  A} = f_{\overline{A}_{\H} \vline A} \, f_{\left(\widetilde{A}_\C,  \, R_\C \right) \vline A}.$
Note that (i) $\widetilde{A}_\C$ and $R_\C$ are independent of the honest agents' affine coefficients $A_\H = [\alpha_i, \ i \in \H]$, and (ii) $\overline{A}_{\H}$ is also independent of the corrupted agents' affine coefficients $A_\C = [\alpha_i, ~ i \in \C]$. Thus, 
$f_{\left(\overline{A}_{\H}, \,\widetilde{A}_\C,  \, R_\C \right) \, \vline \,  A} = f_{\overline{A}_{\H} \, \vline \,  A_\H} ~ f_{\left(\widetilde{A}_\C,  \, R_\C \right) \, \vline \,  A_\C}$.
Upon substituting this in~\eqref{eqn:prob_view}, and using~\eqref{eqn:view_final}, we obtain that
\begin{align}
    \view_\C(A) = f_{\overline{A}_{\H} \, \vline \,  A_\H} ~ f_{\left(\widetilde{A}_\C,  \, R_\C \right)\, \vline \,  A_\C}. \label{eqn:view_final_A}
\end{align}
Now, consider an alternate scenario where the agents' collective affine coefficients are $B = [\beta_1, \ldots, \, \beta_n]$, such that $\beta_i = \alpha_i, ~ \forall i \in \C$, and $\sum_{i \in \V}\beta_i = \sum_{i \in \V} \alpha_i$. Using similar arguments as above, we will obtain that
\begin{align}
    \view_\C(B) = f_{\overline{B}_{\H} \, \vline \,  B_\H} ~ f_{\left(\widetilde{B}_\C,  \, R_\C \right) \, \vline \,  B_\C}  \label{eqn:view_final_B}
\end{align}
where $\overline{B}_{\H}$, $B_\H$, $\widetilde{B}_\C$ and $B_\C$ are the counterparts of $\overline{A}_{\H}$, $A_\H$, $\widetilde{A}_\C$ and $A_\C$, respectively. 

Using the additive property of KL-divergence~\cite{cover2012elements}, from~\eqref{eqn:view_final_A} and~\eqref{eqn:view_final_B} we obtain that 
\begin{align}
   & \D_{KL}\left(\view_\C(A), \, \view_\C(B)\right) = \D_{KL} \left(f_{\overline{A}_{\H} \, \vline \, A_\H}, \, f_{\overline{B}_{\H} \, \vline \, B_\H}\right)  \nonumber \\
    & \qquad \qquad + \D_{KL}\left(f_{\left(\widetilde{A}_\C,  \, R_\C \right) \, \vline \, A_\C}, \, f_{\left(\widetilde{B}_\C,  \, R_\C \right) \, \vline \, B_\C}\right). \label{eqn:view_kl}
\end{align}
As the affine coefficients $A_C$ and $B_C$ are identical to each other, we get from~\eqref{eqn:alpha_alpha}, the conditional probability distributions $f_{\left(\widetilde{A}_\C,  \, R_\C \right) \, \vline \, A_\C}$ and $f_{\left(\widetilde{B}_\C,  \, R_\C \right) \, \vline \, B_\C}$ are equivalent. Therefore, 
$\D_{KL} \left(f_{\left(\widetilde{A}_\C,  \, R_\C \right) \, \vline \, A_\C}, \, f_{\left(\widetilde{B}_\C,  \, R_\C \right) \, \vline \, B_\C} \right) = 0.$
Upon substituting this in~\eqref{eqn:view_kl} we obtain that 
\begin{align}
    \D_{KL}\left(\view_\C(A), \, \view_\C(B)\right) = \D_{KL} \left(f_{\overline{A}_{\H} \, \vline \, A_\H}, \, f_{\overline{B}_{\H} \, \vline \, B_\H}\right). \label{eqn:final_kl_view}
\end{align}
Let $\G_\H = \left( \H, \E_\H \right)$ be the residual honest graph, and let $\L_{\H}$ denote the graph-Laplacian of $\G_\H$. As we assume that $\C$ is not a vertex cut of $\G$, $\G_\H$ is connected. Therefore, substituting from Lemma~\ref{lem:priv_1} in~\eqref{eqn:final_kl_view} we obtain that 
\[\D_{KL} \left(\view_\C(A) , \,  \view_\C(B) \right) \leq \frac{1}{4\sigma^2\underline{\mu}(\L_\H)} \norm{A_\H - B_\H}^2.\]
As $A_\C = B_\C$, $\norm{A_\H - B_\H}^2 = \norm{A - B}^2$.  $\hfill \blacksquare$

\section{Numerical Simulation}
\label{sec:sim}
In this section, we present a numerical simulation of the \af{FS} protocol. We consider a network of 3 agents, $\{1, \, 2, \, 3\}$, connected in a {\em complete} graph. The agents' private local costs are $h_1(x) = x^2 + x, \, h_2(x) = x^2 + 2x, \text{ and } h_3(x) = x^2 + 3x$, where $x \in [-100,\,100]$. Thus, $A = [\alpha_1, \, \alpha_2, \, \alpha_3] = [1, \, 2, \, 3]$.
For computing the effective cost functions, defined in~\eqref{eqn:eff_cost}, the agents use $\sigma = 1$ in phase I. In phase II, we simulate the \af{DGD} on the effective cost functions. 
The absolute differences of the agents' local estimates from the minimizer of the aggregate cost is plotted in Fig.~\ref{fig:estimates}, for both the \af{FS} protocol and the conventional \af{DGD} algorithm, to show convergence. 


\begin{figure}[t!]
\centering
\centering \includegraphics[width=0.35\textwidth]{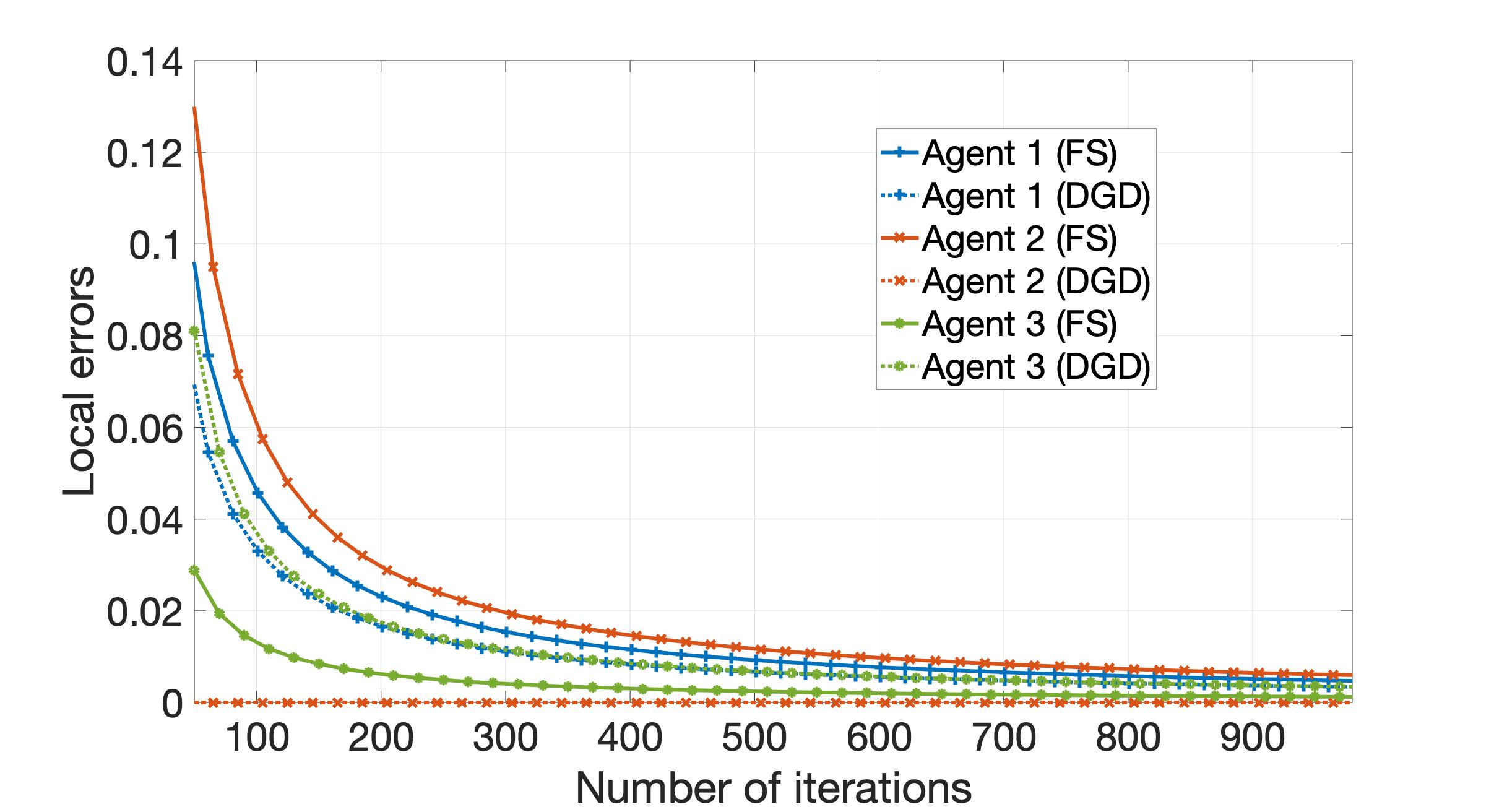} 
\caption{The agents' local errors from the optimizer~\eqref{eqn:opt} converges to zero.}
\label{fig:estimates}
\end{figure}
\begin{figure}[t!]
\centering
\centering \includegraphics[width=0.38\textwidth]{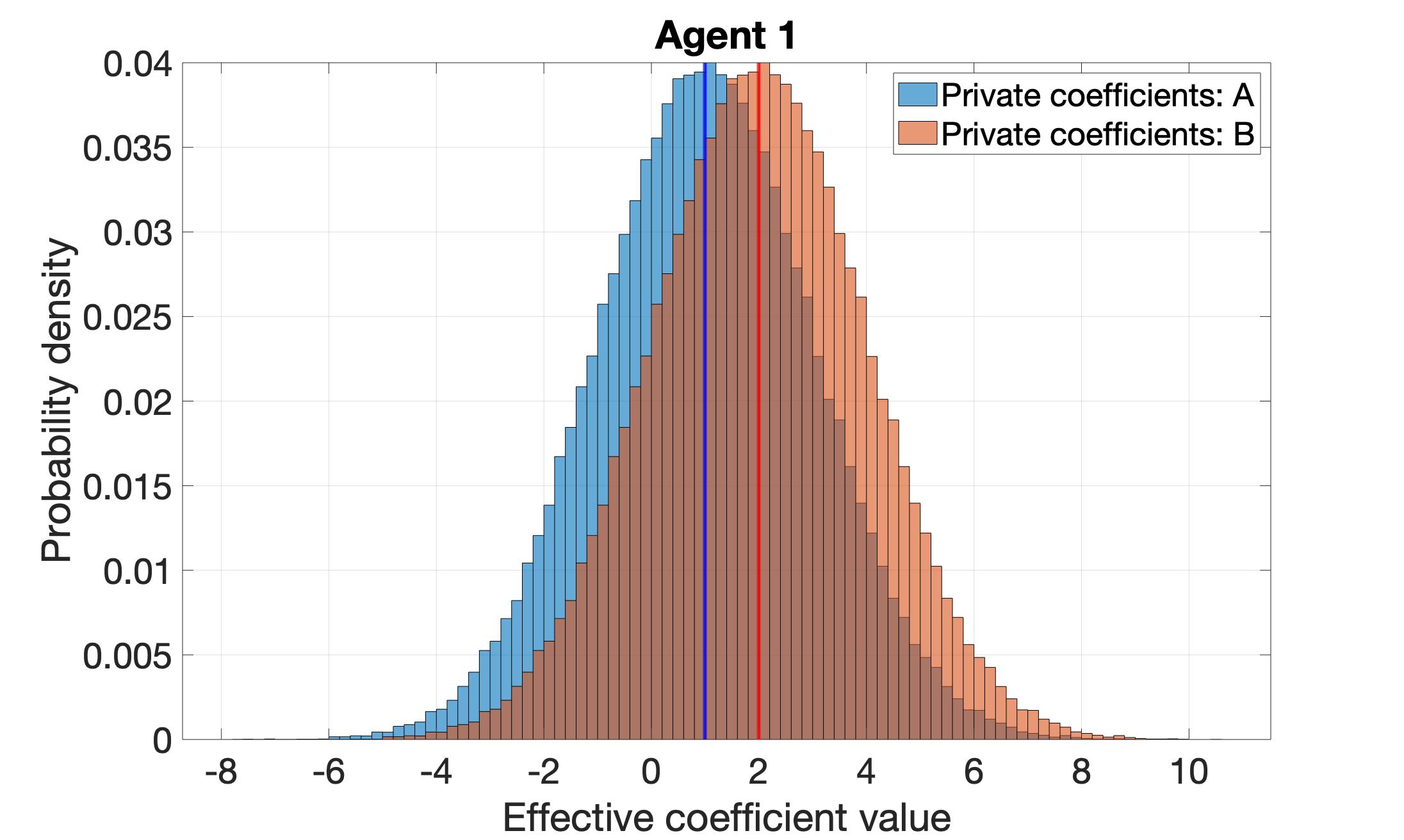} 
\caption{The p.d.f's of agent 1's affine coefficients computed numerically for the two scenarios when the agents' private coefficients are $A$ and $B$.}
\label{fig:prob}
\end{figure}

We assume agent $3$ to be corrupted by a passive adversary, i.e., $\C = \{3\}$ and $\H = \{1, \, 2\}$. We consider an alternate scenario where agents' affine coefficients are given by $B = [\beta_1, \, \beta_2, \, \beta_3] = [2, 1, 3]$. 
Note that $\alpha_3 = \beta_3$ and $\sum_{i = 1}^3 \beta_i = \sum_{i = 1}^3 \alpha_i = 6$. We simulate $100,000$ executions of the \af{FS} protocol for both scenarios. The p.d.f's of agent 1's effective affine coefficients generated in phase I for both the scenarios are shown in Fig.~\ref{fig:prob}. To compute the value of $\D_{KL}(\view_\C(A), \, \view_\C(B))$, we first numerically approximate $p_A$ and $p_B$, the respective conditional p.d.f.s of the effective coefficients $\left[\overline{\alpha}_1, \, \overline{\alpha}_2 \right]$ and $\left[\overline{\beta}_1, \, \overline{\beta}_2 \right]$ (defined by~\eqref{eqn:over_alpha}) given the agents coefficients $A$ and $B$, using the MATLAB's `$\mathsf{fitdist}$' function. Note that, owing to~\eqref{eqn:final_kl_view}, $\D_{KL}(\view_\C(A), \, \view_\C(B)) = \D_{KL}(p_A, \, p_B)$. 
We obtain that $p_A$ and $p_B$ are Gaussian distributions with mean values $\mu_A = [1.00, 2.00]$ and $\mu_B = [2.00, 1.00]$, respectively, and an identical covariance matrix $\Sigma = [2.00,  -2.00; -2.00, 2.00]$. Thus, $\D_{KL}(p_A, \, p_B) = 0.5 (\mu_A - \mu_B) \Sigma^\dagger (\mu_A - \mu_B)^T = 0.25.$
This matches the theoretical bound computed by substituting $\underline{\mu}(\L_\H) = 2$, $\sigma = 1$, and $\norm{A - B}^2 = 2$ in Theorem~\ref{thm:priv_2}.
\section{Concluding Remarks}
\label{sec:dis}

We have presented a theoretical privacy analysis for the Function Sharing or \af{FS} protocol, a distributed optimization protocol proposed in~\cite{gade2018private} for protecting privacy of agents' costs against a passive adversary that corrupts some of the agents in the network. We have shown that the \af{FS} protocol 
preserves the statistical privacy of the polynomial terms of the honest agents' costs if the corrupted agents do not constitute a vertex cut of the network. If the network has $(t+1)$-connectivity then the statistical privacy of the \af{FS} protocol holds true against all passive adversaries that corrupt at most $t$ agents.

\bibliographystyle{IEEEtran}
\bibliography{references_optimization}

\end{document}